\documentclass{llncs}
\usepackage{llncsdoc}
\usepackage{graphicx}
\usepackage{amsmath}
\usepackage{xspace}
\usepackage{times}
\usepackage{enumerate}
\usepackage{algorithm}
\usepackage{algpseudocode}
\usepackage{hyperref}
\usepackage{sober}
\usepackage[margin=1.5in]{geometry}

\newtheorem{fact}[theorem]{Fact}

\newcommand{\broadcast}{{\it BROADCAST}\xspace}
\newcommand{\compute}{{\it COMPUTE}\xspace}
\newcommand{\computecircuit}{{\it COMPUTE-CIRCUIT}\xspace}
\newcommand{\triggercheck}{{\it TRIGGER-CHECK}\xspace}
\newcommand{\advancedcheck}{{\it CHECK}\xspace}
\newcommand{\elect}{{\it ELECT}\xspace}
\newcommand{\request}{{\it REQUEST}\xspace}
\newcommand{\recompute}{{\it RECOMPUTE}\xspace}
\newcommand{\resendresult}{{\it RESEND-RESULT}\xspace}
\newcommand{\update}{{\it UPDATE}\xspace}
\newcommand{\investigate}{{\it INVESTIGATE}\xspace}
\newcommand{\markconflicts}{{\it MARK-IN-CONFLICTS}\xspace}

\newcommand{\circuitsize}{m}
\newcommand{\logstarcircuitsize}{\log^*{\circuitsize}}
\newcommand{\chkprobinv}{(\logstarcircuitsize)^2}
\newcommand{\asymchkrows}{O(\logstarcircuitsize)}
\newcommand{\realchkrows}{16\logstarcircuitsize}

\newcommand{\receiver}{{\bf r}\xspace}

\newcommand{\inactivenode}{inactive}
\newcommand{\activenode}{active}
\newcommand{\neutralnode}{neutral}

\newcommand{\paperTitle}{Self-Healing Computation}

\newcommand{\paperTitleAbbr}{Self-Healing Computation}

\newcommand{\authorsAbbr}{G. Saad and J. Saia}

\date{}

\begin{document}

\mainmatter  

\author{George Saad \and Jared Saia}

\institute{Department of Computer Science, University of New Mexico,\\
\email{\{saad,saia\}@cs.unm.edu}
}

\title{\paperTitle
\footnote{
This research is partially supported by NSF Award SATC 1318880.
}
}

\titlerunning{\paperTitleAbbr}

\authorrunning{\authorsAbbr}
\tocauthor{George Saad and  Jared Saia}

\toctitle{\paperTitleAbbr}
\maketitle

\begin{abstract}
In the problem of reliable multiparty computation (RC), there are $n$
parties, each with an individual input, and the parties want to
jointly compute a function $f$ over $n$ inputs.  The problem is
complicated by the fact that an omniscient adversary controls a hidden
fraction of the parties.

We describe a self-healing algorithm for this problem.  In particular, for a fixed function $f$, with $n$ parties and $m$ gates, we describe how to perform RC repeatedly as the inputs to $f$ change.  Our algorithm maintains the following
properties, even when an adversary controls up to $t \leq (\frac{1}{4}
- \epsilon) n$ parties, for any constant $\epsilon >0$.
First, our algorithm performs each reliable computation with the following amortized resource costs: $O(m + n \log n)$ messages, $O(m + n \log n)$ computational operations, and $O(\ell)$ latency, where $\ell$ is the depth of the circuit that computes $f$.
Second, the expected total number of corruptions is $O(t \chkprobinv )$, after which the adversarially controlled parties are effectively quarantined so that they cause no
more corruptions.
\end{abstract}

\section{Introduction}
How can we protect a network against adversarial attack?  A traditional approach provides robustness through redundant components.  If one component is attacked, the remaining components maintain functionality.  Unfortunately, this approach incurs significant resource cost, even when the network is not under attack.

An alternative approach is self-healing, where a network automatically recovers from attacks.  Self-healing algorithms expend additional resources only when it is necessary to repair from attacks. 

In this paper, we describe self-healing algorithms for the problem of \textit{reliable multiparty computation (RC)}.  In the RC problem, there are $n$
parties, each with an individual input, and the parties want to
jointly compute a function $f$ over $n$ inputs.  A hidden $1/4$-fraction of the parties are controlled by an omniscient Byzantine adversary.  A party that is controlled by the adversary is said to be \emph{bad}, and the remaining parties are said to be \emph{good}.  Our goal is to ensure that all good parties learn the output of $f$. \footnote{Note that RC differs from secure multiparty computation (MPC) only in that there is no requirement to keep inputs private.}

RC abstracts many problems that may occur in high-performance computing, sensor networks, and peer-to-peer networks.  For example, we can use RC to enable performance profiling and system monitoring, compute order statistics, and enable public voting. 

Our main result is an algorithm for RC that 1) is asymptotically optimal in terms of total messages and total computational operations; and 2) limits the expected total number of corruptions.  Ideally, each bad party would cause $O(1)$ corruptions; in our algorithm, each bad party causes an expected $O(\chkprobinv)$ corruptions.

\subsection{Our Model}

We assume a static Byzantine adversary that takes over $t \leq (\frac{1}{4} - \epsilon) n$ parties before the algorithm begins, for any constant $\epsilon >0$. 
As mentioned previously, parties that are compromised by the adversary are called \emph{bad}, and the remaining parties are \emph{good}. The bad parties may arbitrarily deviate from the protocol, by sending no messages, excessive numbers of messages, incorrect messages,  or any combination of these.  The good parties follow the protocol.  We assume that the adversary knows our protocol, but is unaware of the random bits of the good nodes.  We make use of a public key cryptography scheme, and thus assume that the adversary is computationally bounded.

We assume a partially synchronous communication model. Any message sent from one good node to another good node requires at most $h$ time steps to be sent and received, and the value $h$ is known to all nodes.  However, we allow the adversary to be \emph{rushing}: the bad nodes receive all messages from good nodes in a round before sending out their own messages.  We further assume that each party has a unique ID.  We say that party $p$ has a link to party $q$ if $p$ knows $q$'s ID and can thus directly communicate with node $q$.

In the reliable multiparty computation problem, we assume that the function $f$ can be implemented with an arithmetic circuit over $m$ gates, where each gate has two inputs and at most two outputs.\footnote{We note that any gate of any fixed in-degree and out-degree can be converted into a fixed number of gates with in-degree $2$ and out-degree at most $2$.}  
For simplicity of presentation, we focus on computing a single function multiple times (with changing inputs).  However, we can also compute multiple functions with our algorithm.

\subsection{Our Result}
We describe an algorithm, \compute, to efficiently solve reliable multiparty computation.  Our main result is summarized in the following theorem. 

\begin{theorem}\label{thm:corruptions}
Assume we have $n$ parties providing inputs to a function $f$ that can be computed by an arithmetic circuit with depth $\ell$ and containing $m$ gates.  Then \compute solves RC and has the following properties.
\begin{itemize}
\item In an amortized sense,\footnote{In particular, if we call \compute $\mathcal{L}$ times, then the expected total number of messages sent will be $O(\mathcal{L} (m + n\log n) + t(m \log^{2} n))$.  Since $t$ is fixed, for large $\mathcal{L}$, the expected number of messages per \compute is $O(m + n\log{n})$.  The result for computational operations is similar.} any execution of \compute requires the following expected costs:
\begin{itemize}
\item $O(m + n \log n)$ messages sent by all parties,
\item $O(m + n \log n)$ computational operations performed by all parties, and
\item $O(\ell)$ latency.
\end{itemize}
\item The expected total number of times \compute returns a corrupted output is $O(t \chkprobinv)$.
\end{itemize}
\end{theorem}

\subsection{Technical Overview}
Our algorithms make critical use of quorums and a quorum graph. 

\smallskip
\noindent
\textbf{Quorums and the Quorum Graph:}  
We define a quorum to be a set of $\Theta(\log{n})$ parties, of which at most $1/4$-fraction are bad.
Many results show how to create and maintain a network of quorums~\cite{FS2,hildrum2003,NW03,scheideler:how,FSY,AS4,king2011load}.  All of these results maintain what we will call a \emph{quorum graph} in which each vertex represents a quorum.
The properties of the quorum graph are:
1) each party is in $\Theta(\log n)$ quorums; 
2) for any quorum $Q$, any party in $Q$ can communicate directly to any other party in $Q$; and 
3) for any quorums $Q$ and $Q'$ that are connected in the quorum graph, any party in $Q$ can communicate directly with any party in $Q'$ and vice versa.
Moreover, we assume that for any two parties $x$ and $y$ in a quorum, $x$ knows all quorums that $y$ is in.

\begin{figure}[t]
\centerline{\includegraphics[scale=0.4,natwidth=475,natheight=318]{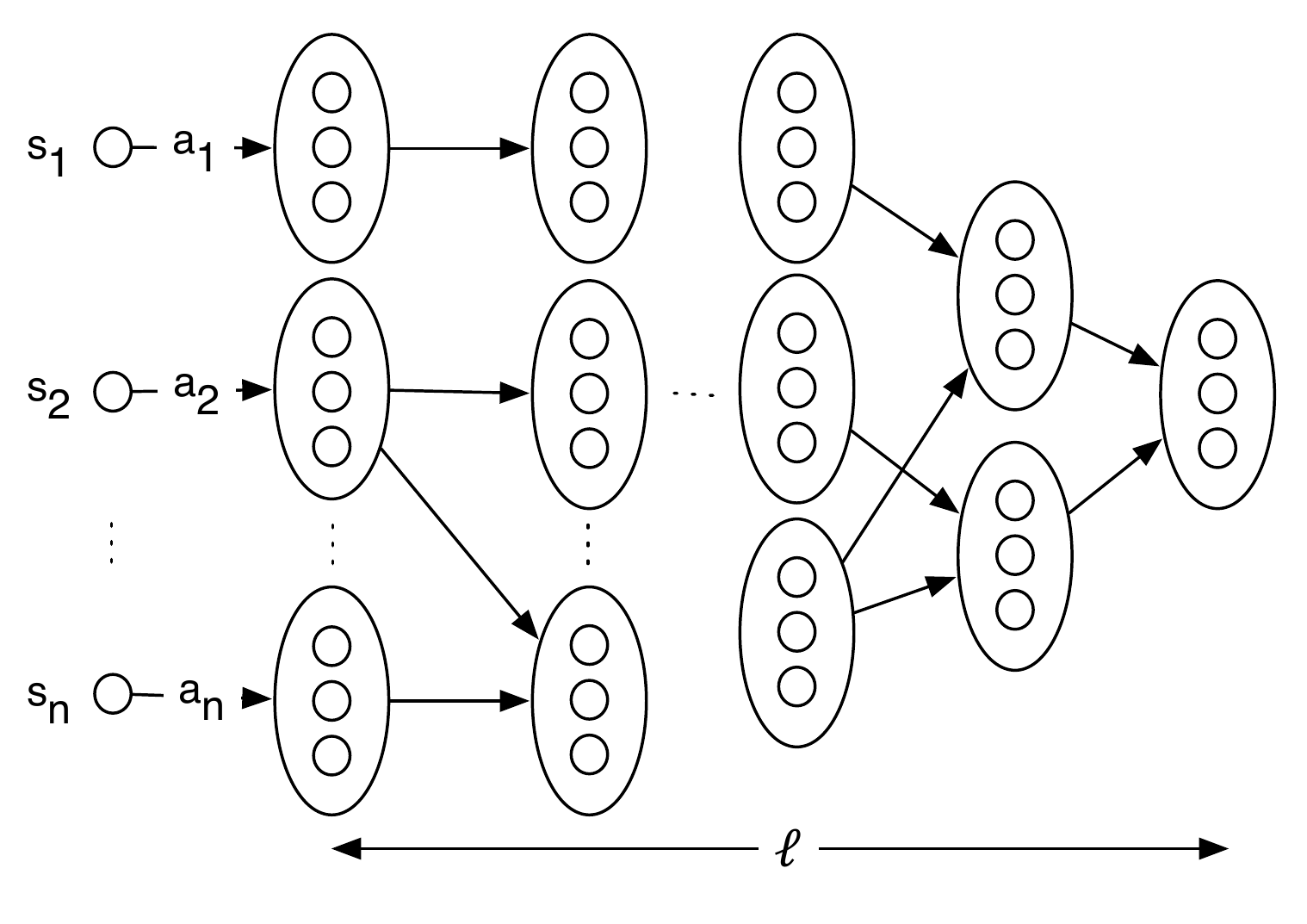}}
\caption{Quorum-Based Circuit}
\label{fig:quorumgraph}
\end{figure}

\smallskip
\noindent
\textbf{Computing with Quorums:}  We maintain a quorum graph with $m+n$ nodes: $m$ nodes for the gates of the circuit and $n$ nodes for the inputs of the parties.  The input nodes are connected to the gates using these inputs, and the gate nodes are connected as in the circuit.  Quorums are mapped to nodes in this quorum graph as described above.  See Figure \ref{fig:quorumgraph}.  Intuitively, the computation is performed from the left to the right, where the input quorums are the leftmost quorums and the output quorum is the rightmost one.

A correct but inefficient way to solve RC is as follows.  Each party $s_i$ sends its input to all parties of the appropriate input quorum.  Then the computation is performed from left to right. All parties in each quorum compute the appropriate gate operation on their inputs, and send their outputs to all parties in the right neighboring quorums via all-to-all communication.  At the next level,
all parties in each quorum take the majority of the received messages in order to determine the correct input for their gate.  At the end, the parties in the rightmost quorum will compute the correct output of the circuit.  They then forward this output back from right to left through the quorum graph using the same all-to-all communication and majority filtering. 

Unfortunately, this naive algorithm requires $O(m \log^{2} n)$ messages and $O(m\log{n})$ computational operations.  Our main goal is to remove the logarithmic factors.  \footnote{We note that such asymptotic improvements can be significant for large networks.  For example, if $n = 64{,}000$, then we would expect our algorithm to reduce message costs by a factor of $\log^2 n = 255$.}

\smallskip
\noindent
\textbf{Leaders and the \advancedcheck Algorithm:}  A more efficient approach is for each quorum to have a leader, and for this leader to receive inputs, perform gate computations, and send off the output.  Unfortunately, a single bad leader can corrupt the entire computation.

To address this issue, we create \advancedcheck (Section~\ref{app:check}).  This algorithm determines if there has been a corruption, and if so, returns at least one pair of parties that are in \emph{conflict}.  Informally, we say that a pair of parties are in conflict if they each accuse the other of malicious behavior.  In such a situation, we know that at least one party in the pair is bad.  
Our approach is to mark both parties in each conflicting pair, and then to forbid any marked party from being a leader of a quorum.
\footnote{A technical point is that we may need to unmark all parties in a quorum if too many parties in that quorum become marked.  However, a potential function argument (Lemma~\ref{l:numc}) shows that it is still the case that after $O(t)$ markings, all bad parties will be marked.}

The basic idea of \advancedcheck is to have multiple rounds where in each round, a new party is selected independently at random from each quorum.  We call these parties the \emph{checkers}.  
For convenience of presentation, we will refer to the leaders as the checkers for round $0$.
For each round $i\geq1$, all $i$ checkers at gate $g$: 1) receive inputs to $g$ from all $i$ checkers at input gates for $g$; 2) compute the gate output for $g$ based on these inputs; and 3) send this output to all $i$ checkers at each output gate for $g$.  If a checker ever receives inconsistent inputs, it calls \update (Section~\ref{sec:update}), which will return at least one pair of parties that are in conflict.  Unfortunately, waiting until a round where each gate has had at least one good checker would require $O(\log n)$ rounds.

To do better, we use the following approach.  Let $G$ be the quorum graph as defined above and let the checkers be selected as above.  Call a subgraph of $G$ bad in a given round if all checkers in the nodes of that subgraph are bad in that round.  Any corruption in the first round must occur in a bad subgraph.  Moreover, for $i$ rounds of \advancedcheck to fail to find a corruption, there must be nesting levels of bad subgraphs in $G$ in each of those $i$ rounds.

When \advancedcheck elects a good checker at a quorum, it is as if it is removing the node associated with the quorum from the quorum graph.
Thus, we can view \advancedcheck as repeatedly removing nodes from increasingly smaller subgraphs of $G$ until no nodes remain, at which the corruption is detected.  
A key lemma (Lemma \ref{lem:shrinkingDeceptionDAG}) shows that for any directed acyclic graph (DAG) with $m$ nodes and maximum degree $4$, when each node is deleted with probability $3/4$, the probability that a connected subgraph of size $\Omega(\log{m})$ survives is at most $1/2$.  Using this lemma, we can show that \advancedcheck requires only $\asymchkrows$ rounds to detect a corruption with constant probability.\footnote{This probability can be made arbitrarily close to $1$ by adjusting the hidden constant in the $\asymchkrows$ rounds.}

Since \advancedcheck requires $O((m+n\log{n})\chkprobinv)$ messages, we can call it with probability $1/\chkprobinv$ and obtain asymptotically optimal resource costs for the RC problem, while incurring an expected $O(t \chkprobinv)$ corruptions.

\subsection{Related Work}

Our results are inspired by recent work on self-healing algorithms. Early work of \cite{Frisanco:1997,Iraschko:1998,Murakami:1997,Caenegem:1997,Xiong:1999} discusses different restoration mechanisms to preserve network performance by adding capacity and rerouting traffic streams in the presence of node or link failures. This work presents mathematical models to determine global optimal restoration paths, and provides methods for capacity optimization of path-restorable networks.

More recent work~\cite{boman2006brief,saia2008picking,hayes2008forgiving,hayes2009forgiving,pandurangan2011xheal,sarma2011edge} considers models where the following process repeats indefinitely: an adversary deletes some nodes in the network, and the algorithm adds edges.  The algorithm is constrained to never increase the degree of any node by more than a logarithmic factor from its original degree.  In this model, researchers have presented algorithms that ensure the following properties: the network stays connected and the diameter does not increase by much~\cite{boman2006brief,saia2008picking,hayes2008forgiving}; the shortest path between any pair of nodes does not increase by much~\cite{hayes2009forgiving}; expansion properties of the network are approximately preserved~\cite{pandurangan2011xheal}; and keeping network backbones densely connected~\cite{sarma2011edge}.

This paper particularly builds on~\cite{KSS:2013}.  That paper describes self-healing algorithms that provide reliable communication, with a minimum of corruptions, even when a Byzantine adversary can take over a constant fraction of the nodes in a network.  While our attack model is similar to~\cite{KSS:2013}, reliable \textit{computation} is more challenging than reliable communication, and hence this paper requires a significantly different technical approach.  Additionally, we improve the fraction of bad parties that can be tolerated from $1/8$ to $1/4$.

Reliable multiparty computation (RC) is closely related to the problem of secure multiparty computation (MPC) which has been studied extensively for several decades (see e.g.~\cite{Yao:1982,Beaver:1992,Ben-Or:1988,Rabin:1989} or the recent book~\cite{prabhakaran2013secure}).  RC is simpler than MPC in that it does not require inputs of the parties to remain private.  Our algorithm for RC is significantly more efficient than current algorithms for MPC, which require at least polylogarithmic blowup in communication and computational costs in order to tolerate a Byzantine adversary.  We reduce these costs through our self-healing approach, which expends additional resources only when corruptions occur, and is able to ``quarantine" bad parties after $O(t \chkprobinv)$ corruptions.

\subsection{Organization of Paper}
The rest of this paper is organized as follows.  In Section~\ref{sec:algorithms}, we describe our algorithms.
The analysis of our algorithms is shown in Section \ref{sec:analysis}.
Finally, we conclude and describe problems for future work in Section~\ref{sec:conclusion}.

\section{Our Algorithms}
\label{sec:algorithms}
In this section, we describe our algorithms: \compute, \computecircuit, \advancedcheck and \update. 

Our algorithms aim at detecting corruptions and marking the bad parties.  Parties that are marked are not allowed to participate in the computation process.  Initially, all parties are unmarked.

Recall that there are $n$ parties, each provides an input to an input quorum, $Q_i$, for $1 \leq i \leq n$; and then the computation is performed through $m$ quorums $Q_j$, for $n+1 \leq j \leq m+n$. The result is produced at an output quorum $Q_{m+n}$.

Before discussing our main \compute algorithm, we describe that when a party $x$ broadcasts a message $msg$, signed by a quorum $Q$, to a set of parties $S$, it calls $\broadcast (msg, Q, S)$.

\subsection{\broadcast}

In \broadcast (Algorithm \ref{a:broadcast}), we use threshold cryptography to avoid the overhead of Byzantine Agreement. In a ($\eta, \eta'$)-threshold cryptographic scheme, a private key is distributed among $\eta$ parties in such a way that 1) any subset of more than $\eta'$ parties can jointly reassemble the key; and 2) no subset of at most $\eta'$ parties can recover the key.
The private key can be distributed using a \emph{Distributed Key Generation} (DKG) protocol\cite{Kate:2009}.  DKG generates the public/private key shares of all parties in every quorum.  The public key of each quorum is known to all parties in the quorum, and to all parties in all neighboring quorums in the circuit.


\begin{algorithm}
\caption{\textsc{BROADCAST}$(msg, Q, S)$ \Comment{A party $x$ calls this procedure in order to send a message $msg$, signed by quorum $Q$, to a set of parties $S$.}}\label{a:broadcast}
\begin{algorithmic}[1]
\State Party $x$ sends message $msg$ to all parties in $Q$.
\State Each party in $Q$ signs $msg$ by its private key share to obtain its message share.
\State Each party in $Q$ sends its message share back to party $x$.
\State Party $x$ interpolates at least $\frac{3|Q|}{4}$ message shares to obtain a signed-message of $Q$.
\State Party $x$ sends this signed-message to all parties in $S$.
\end{algorithmic}
\end{algorithm}

In particular, we use a ($|Q|$, $\frac{3|Q|}{4} -1$)-threshold scheme, where $|Q|$ is the quorum size.  A party $x$ calls $\broadcast$ in order to send a message $msg$ to all parties in $S$ so that: 1) at least $3/4$-fraction of the parties in quorum $Q$ have received the same message $msg$; 2) they agree upon the content of $msg$;  and 3) they give a permission to $x$ to broadcast this message.

Any call to \broadcast has $O(\log{n} + |S|)$ messages, $O(\log{n})$ computational operations for signing the message $msg$ by $O(\log{n})$ parties in $Q$ and $O(1)$ latency.

\subsection{\compute}
Now we describe our main algorithm, \compute (Algorithm~\ref{a:compute}), which calls \computecircuit (Algorithm \ref{a:compute-circuit}).
In \computecircuit, the $n$ parties broadcast their inputs to the input quorums. The input quorums forward these inputs to a circuit of $m$ leaders in order to perform the computation and provide the result to the output quorum. Then this result is sent through the same circuit of leaders back to all parties.
Recall that a leader of a quorum, is a party in this quorum, that is: 1) a representative of all parties in the quorum; and 2) it is known to all parties in the quorum and neighboring quorums.
We assume that all parties provide their inputs to the circuit in the same round.
\begin{algorithm}
\caption{\textsc{COMPUTE}\Comment{performs a reliable computation and sends the result reliably to all parties.}}
\label{a:compute}
\begin{algorithmic}[1]
\State COMPUTE-CIRCUIT \Comment{computes and sends back the result through a circuit of leaders.}
\State TRIGGER-CHECK \Comment{The output quorum triggers \advancedcheck with probability $1/\chkprobinv$.}
\end{algorithmic}
\end{algorithm}

In the presence of an adversary, \computecircuit is vulnerable to corruptions.  
Thus, \compute calls \triggercheck (Algorithm \ref{a:triggercheck}), in which the parties of the output quorum decide together, to trigger \advancedcheck (Algorithm \ref{a:check}) with probability $1/\chkprobinv$, using secure multiparty computation (MPC) \cite{Rabin:1989}.
\advancedcheck is triggered in order to detect if a computation was corrupted in the last call to \computecircuit, with probability at least $1/2$.

\begin{algorithm}
\caption{\textsc{COMPUTE-CIRCUIT}\Comment{performs a computation through a circuit of leaders producing a result at the output quorum; and it sends back result through same circuit to all parties.}}
\label{a:compute-circuit}
\begin{algorithmic}[1]
\For{$i = 1, \ldots, n$}\Comment{provides the inputs to the circuit}
	\State Party $s_i$ broadcasts its input $a_i$ to all parties in quorum $Q_i$.
	\State All parties in $Q_i$ send $a_i$ to the leaders of the right neighboring quorums of $Q_i$.
\EndFor
\For{$i = n+1, \ldots, m+n-1$} \Comment{performs the computation.}
	\ForAll{$j: i< j \leq m+n$ and $(Q_{i}, Q_{j}) \in Circuit$}
		\If{leader $q_{i} \in Q_{i}$ receives all its inputs}
			\State $q_i$ performs an operation on its inputs producing an output, $b_i$.
			\State $q_{i}$ sends $b_i$ to leader $q_{j} \in Q_{j}$.
		\EndIf
	\EndFor
\EndFor
\If{leader $q_{m+n} \in Q_{m+n}$ receives all its inputs} \Comment{produces and broadcasts result to output quorum.}
	\State $q_{m+n}$ performs an operation on its inputs producing an output, $b_{m+n}$.
	\State $q_{m+n}$ broadcasts $b_{m+n}$ to all parties in $Q_{m+n}$.
\EndIf

\For{$i = m+n, \ldots, n+2$} \Comment{sends back the result to the leftmost leaders.}
	\ForAll{$j: n+1\leq j< i$ and $(Q_{j}, Q_{i}) \in Circuit$}
		\State Leader $q_{i} \in Q_i$ sends $b_{m+n}$ to leader $q_{j} \in Q_{j}$.
	\EndFor
\EndFor
\For{$i = 1, \ldots, n$}\Comment{sends result to all parties after broadcasting it to the input quorums.}
	\State The leaders of the right neighboring quorums of $Q_i$ broadcast $b_{m+n}$ to all parties in $Q_{i}$.
	\State All parties in $Q_i$ send $b_{m+n}$ to sender $s_i$.
\EndFor
\end{algorithmic}
\end{algorithm}

\begin{algorithm}
\caption{\textsc{TRIGGER-CHECK} \Comment{The parties of the output quorum $Q_{m+n}$ trigger \advancedcheck with probability $1/\chkprobinv$.}}
\label{a:triggercheck}
\begin{algorithmic}[1]
	\State Each party in $Q_{m+n}$ chooses an input: a real number uniformly distributed between $0$ and $1$.
	\State The parties of $Q_{m+n}$ perform MPC to find the output, $prob$, which is the sum of all their inputs divided by $|Q_{m+n}|$.
	\If{$prob \leq 1/\chkprobinv$}
		\State CHECK
	\EndIf
\end{algorithmic}
\end{algorithm}

Unfortunately, while \advancedcheck can determine if a corruption occurred, it does not specify the location where the corruption occurred.  
Thus, if \advancedcheck detects a corruption, \update (Algorithm~\ref{a:update}) is called.   
When \update is called, it identifies two neighboring quorums $Q$ and $Q'$ in the circuit, such that at least one pair of parties in these quorums is in conflict and at least one party in this pair is bad. 
Then quorums $Q$ and $Q'$ mark these parties and notify all other quorums that these parties are in.
All quorums in which these parties are notify their neighboring quorums.
For each pair of leaders that is in conflict, their quorums elect a new pair of unmarked leaders uniformly at random.
If $(1/2 - \gamma)$-fraction of parties in any quorum have been marked, for any constant $\gamma > 0$, e.g., $\gamma = 0.01$, they are set unmarked in all their quorums and neighboring quorums.

Moreover, we use \broadcast in \computecircuit and \advancedcheck in order to handle any accusation issued in \update against the parties that provide the inputs to the input quorums, or those that receive the result in the output quorum.

Our model does not directly consider concurrency.  
In a real system, concurrent executions of \compute that overlap at a single quorum may allow the adversary to achieve multiple corruptions at the cost of a single marked bad party.  However, this does not effect correctness, and, in practice, this issue can be avoided by serializing concurrent executions of \compute.  For simplicity of presentation, we leave the concurrency aspect out of this paper.

\subsection{\advancedcheck}\label{app:check}

In this section, we describe \advancedcheck algorithm, which is stated formally as Algorithm~\ref{a:check}. In this algorithm, we make use of subquorums, where a subquorum is a subset of unmarked parties in a quorum.
Let $U_{k}$ be the set of all unmarked parties in quorum $Q_{k}$, for $1\leq k \leq m+n$.

\begin{algorithm}
\caption{$\textsc{CHECK}$ \Comment{Party $\receiver$ calls \advancedcheck to check for corruptions.}}
\label{a:check}
\begin{algorithmic}[1]
	\State Let subquorums, $S^1_{j}$, $S^2_{j}$ and $S^3_j$, be initially empty, for all $n+1\leq j \leq m+n$.
	\For{$i \gets 1, \ldots, \realchkrows$}
		\State ELECT($Q_{m+n}$) \Comment{elects a party $\receiver \in Q_{m+n}$.}
		\State Party $\receiver$ constructs $R_1$, $R_2$ and $R_3$ to be three $m$ by $m'$ arrays of random numbers, where $m'$ is the maximum size of any quorum. 
Note that $R_1[k, k']$, $R_2[k, k']$ and $R_3[k, k']$ are three uniformly random numbers between $1$ and $k'$, for $1\leq k \leq m$ and $1\leq k' \leq m'$.
		\State REQUEST($R_1, R_2, i$) \Comment{$\receiver$ requests all senders to recompute.}
		\State RECOMPUTE($R_2, i, \receiver$) \Comment{recomputes producing the result, $b_{m+n}$, at $\receiver$.}
		\State RESEND-RESULT($R_3, b_{m+n}, i$) \Comment{$\receiver$ sends back $b_{m+n}$ to all parties.}
	\EndFor
\end{algorithmic}
\textbf{Note that:} for each algorithm that is called in \advancedcheck, if any party has previously received a message, it verifies this message with all subsequent messages; also if a party receives inconsistent messages or fails to receive an expected message, then it initiates a call to \update. 
\end{algorithm}

\begin{algorithm}
\caption{\textsc{ELECT}$(Q)$ \Comment{Parties in $Q$ elect an unmarked party in $Q$ using MPC.}}
\label{a:elect}
\begin{algorithmic}[1]
	\State Let each party in the set of unmarked parties, $U \subset Q$, is assigned a unique integer from 1 to $|U|$.
	\State Each party in $Q$ chooses an input: an integer uniformly distributed between 1 and $|U|$.
	\State The parties of $Q$ perform MPC to find the output: the sum of all their inputs modulo $|U|$.
	\State The party in $U$ associated with this output number is the elected party.
\end{algorithmic}
\end{algorithm}

\advancedcheck runs for $\asymchkrows$ rounds.
For each round $i$, the parties of the output quorum $Q_{m+n}$ elect an unmarked party $\receiver$ from $Q_{m+n}$ to be in charge of the recomputation in round $i$. This election process is stated formally in \elect (Algorithm \ref{a:elect}).
The elected party $\receiver$ calls \request (Algorithm \ref{a:request}) to send a request through a DAG of subquorums, $S^1_{j}$'s, to the $n$ senders to recompute.
Then the recomputation starts by \recompute (Algorithm \ref{a:recompute}), in which each sender that receives such request provides its input to redo the computation through a DAG of subquorums, $S^2_{j}$'s.
Finally, when $\receiver$ receives the result of the computation, it calls \resendresult (Algorithm \ref{a:resend-result}) in order to send back this result to the senders through a DAG of subquorums $S^3_{j}$'s, for $n+1\leq j \leq m+n$.

\begin{algorithm}
\caption{\textsc{REQUEST}($R_1, R_2, i$) \Comment{$\receiver$ requests $n$ senders through a DAG of subquoums, $S^1_j$'s, for $n+1\leq j \leq m+n$, to redo the computation.}}
\label{a:request}
\begin{algorithmic}[1]
	\State $\receiver$ sets $REQ$ to be a message consisting of $R_1, R_2, i$ and $\receiver$.
	\State $\receiver$ broadcasts $REQ$ to all parties of quorum $Q_{m+n}$.
	\State All parties in $Q_{m+n}$ calculate party, $q_{m+n} \in U_{m+n}$, of index $R_1[m,|U_{m+n}|]$ to be added to $S^1_{m+n}$.
	\State The parties in $Q_{m+n}$ send $REQ$ to the parties of $S^1_{m+n}$.
	\For{$j \gets m+n, \ldots, n+2$} \Comment{forwards $REQ$ from output quorum to input quorums.}
		\ForAll{$k: n+1\leq k<j$ and $(Q_k, Q_j) \in Circuit$} 
			\State All $i$ parties in $S^1_{j}$ calculate party, $q_{k} \in U_{k}$, of index $R_1[k-n,|U_{k}|]$ to be added to $S^1_{k}$.
			\State The parties in $S^1_{j}$ forward $REQ$ and the IDs of all parties in $S^1_{k}$ to party $q_{k}$.
			\State Party $q_{k}$ sends $REQ$ to all the parties in $S^1_{k}$.
		\EndFor
	\EndFor
	\For{$k \gets n, \ldots, 1$}\Comment{Input quorums forward $REQ$ to all senders.}
		\State All $i$ parties in the right neighboring subquorums of $Q_k$ broadcast $REQ$ to all parties in $Q_{k}$.
		\State All parties in $Q_k$ send $REQ$ to sender $s_k$.	
	\EndFor
\end{algorithmic}
\end{algorithm}

Moreover, in \elect($Q$), the parties of quorum $Q$ use MPC to elect an unmarked party uniformly at random from $Q$. Note that at any moment at least half of the unmarked parties in $Q$ are good, thus the elected party is good with probability at least $1/2$. Finally, this election protocol runs in $O(1)$ time, and requires $O(\log^2{n})$ messages and $O(\log{n})$ computational operations.

\begin{algorithm}
\caption{\textsc{RECOMPUTE}($R, i, \receiver$) \Comment{The $n$ senders provide inputs to a DAG of subquorums, $S^2_j$'s, for $n+1\leq j \leq m+n$, in order to recompute producing a result, $b_{m+n}$, at $\receiver$.}}
\label{a:recompute}
\begin{algorithmic}[1]
	\For{each sender $s_j$ that receives $REQ$; $1\leq j \leq n$ and $n+1\leq k,k'\leq m+n$} \Comment{provides the inputs.}
		\State $s_j$ broadcasts its input $a_j$, $i$ and $R$ to all parties in $Q_j$.
		\State All parties in $Q_j$ calculate at most two right neighboring parties, $q_{k}\in U_k$ and $q_{k'} \in U_{k'}$, of indices $R[k-n,|U_{k}|]$ and $R[k'-n,|U_{k'}|]$ to be added to the right neighboring subquorums, $S^2_{k}$ and $S^2_{k'}$.
		\State All parties in $Q_j$ send $a_j$, $i$ and $R$ to the parties in $S^2_{k}$ and the parties in $S^2_{k'}$.
	\EndFor
	\For{$j \gets n+1, \ldots, m+n-1$} \Comment{recomputes}
		\ForAll{$k: j< k \leq m+n$ and $(Q_j,Q_k) \in Circuit$} 
			\If{the new added party, $q_{j} \in S^2_{j}$, receives all its inputs}
				\State $q_{j}$ performs an operation on its inputs producing an output $b_{j}$.
				\State $q_{j}$ sends $b_{j}$, $i$ and $R$ to the parties of $S^2_{j}$.
				\State All $i$ parties in $S^2_{j}$ calculate party, $q_{k} \in U_k$, of index $R[k-n,|U_{k}|]$ to be added to $S^2_{k}$.
				\State The parties in $S^2_{j}$ send $b_{j}$, $i$, $R$ and the IDs of all parties in $S^2_{k}$ to party $q_{k}$.
			\EndIf
		\EndFor
	\EndFor
	\State All $i$ parties in $S_{m+n}$ broadcast $b_{m+n}$, $i$ and $R$ to all parties in $Q_{m+n}$.
	\State All parties in $Q_{m+n}$ send $b_{m+n}$, $i$ and $R$ to party $\receiver$. \Comment{$\receiver$ receives the result.}
\end{algorithmic}
\end{algorithm}

During \advancedcheck, if any party receives inconsistent messages or fails to receive and verify any expected message in any round, it initiates a call to \update.

\begin{algorithm}
\caption{\textsc{RESEND-RESULT}($R, b_{m+n}, i$) \Comment{Party $\receiver$ sends back the result of the computation, $b_{m+n}$, through a DAG of subquorums, $S^3_j$'s, to $n$ senders, for $n+1\leq j \leq m+n$.}}
\label{a:resend-result}
\begin{algorithmic}[1]
	\State $\receiver$ sets $RESULT$ to be a message consisting of $R,  b_{m+n}, i$ and $\receiver$.
	\State $\receiver$ broadcasts $RESULT$ to all parties of quorum $Q_{m+n}$.
	\State All parties in $Q_{m+n}$ calculate party, $q_{m+n} \in U_{m+n}$, of index $R[m,|U_{m+n}|]$ to be added to $S^3_{m+n}$.
	\State The parties in $Q_{m+n}$ send $RESULT$ to the parties of $S^3_{m+n}$.
	\For{$j \gets m+n, \ldots, n+2$}\Comment{forwards result from output quorum to input quorums.}
		\ForAll{$k: n+1\leq k<j$ and $(Q_k, Q_j) \in Circuit$} 
			\State All $i$ parties in $S^3_{j}$ calculate party, $q_{k} \in U_{k}$, of index $R[k-n,|U_{k}|]$ to be added to $S^3_{k}$.
			\State The parties in $S^3_{j}$ forward $RESULT$ and the IDs of all parties in $S^3_{k}$ to party $q_{k}$.
			\State Party $q_{k}$ sends $RESULT$ to all the parties in $S^3_{k}$.
		\EndFor
	\EndFor
	\For{$k \gets n, \ldots, 1$} \Comment{Input quorums forward the result to all senders.}
		\State All $i$ parties in the right neighboring subquorums of $Q_k$ broadcast $RESULT$ to all parties in $Q_k$.
		\State All parties in $Q_{k}$ send $RESULT$ to party $s_k$.
	\EndFor
\end{algorithmic}
\end{algorithm}

\subsection{\update}\label{sec:update}

When a computation is corrupted and \advancedcheck detects this corruption, \update is called. 
The \update algorithm is described formally as Algorithm~\ref{a:update}. 
When \update starts, all parties in each quorum in the circuit are notified.

\begin{algorithm}
\caption{$\textsc{UPDATE}$ \Comment{Party $q' \in Q'$ calls \update after it claims that it detects a corruption.}}
\label{a:update}
\begin{algorithmic}[1]
	\State $q'$ broadcasts the fact that it calls \update along with the messages it has received in this call to \compute to all parties in $Q'$.  
	\State The parties in $Q'$ verify that $q'$ received inconsistent messages before proceeding.
	\State $Q'$ notifies all quorums in the circuit via all-to-all communication that \update is called.
	\State INVESTIGATE \Comment{investigates all participants to determine corruption locations.}
	\State MARK-IN-CONFLICTS \Comment{marks the parties that are in conflict.}
\end{algorithmic}
\end{algorithm}

The main purpose of \update is to 
1) determine the location in which the corruption occurred; and 
2) mark the parties that are in conflict.

\begin{algorithm}
\caption{\textsc{INVESTIGATE}\Comment{investigates parties to determine corruption locations.}}
\label{a:investigate}
\begin{algorithmic}[1]
	\For{each party, $q$, involved in the last call to \computecircuit or \advancedcheck}
		\State $q$ compiles all messages they have received (and from whom) and they have sent (and to whom) in the last call to \computecircuit or \advancedcheck.
		\State $q$ broadcasts these messages to all parties in its quorum and neighboring quorums.
	\EndFor
\end{algorithmic}
\end{algorithm}

To determine the location in which the corruption occurred, \update calls \investigate (Algorithm \ref{a:investigate}) to investigate the current situation by letting each party involved in \computecircuit or \advancedcheck broadcast all messages they have received or sent.
Then, \update calls \markconflicts (Algorithm \ref{a:mark}) in order to mark the parties that are \emph{in conflict}, where a pair of parties is in conflict if at least one of these parties broadcasted messages that conflict with the messages broadcasted by the other party in this pair.
Note that each pair of parties that is in conflict has at least one bad party.
Recall that if $(1/2 - \gamma)$-fraction of parties in any quorum are marked, for any constant $\gamma>0$, e.g., $\gamma = 0.01$, they are set unmarked. 
Also, for each pair of leaders that get marked, their quorums elect another pair of unmarked leaders.

\begin{algorithm}
\caption{$\textsc{MARK-IN-CONFLICTS}$ \Comment{marks the parties that are in conflict.}}
\label{a:mark}
\begin{algorithmic}[1]
	\For{each pair of parties, $(q_x, q_y)$, that is in conflict*, in quorums $(Q_x, Q_y)$}
		\State party $q_y$ broadcasts a \emph{conflict} message, $\{q_x, q_y\}$, to all parties in $Q_y$.
		\State each party in $Q_y$ forwards $\{q_x, q_y\}$ to all parties in $Q_x$.
		\State all parties in $Q_x$ (or $Q_y$) send $\{q_x, q_y\}$ to the other quorums that has $q_x$ (or $q_y$). 
		\State each quorum has $q_x$ or $q_y$ sends $\{q_x, q_y\}$ to its neighboring quorums.
	\EndFor
	\For{each party $q$ that receives conflict message $\{q_x, q_y\}$}
		\State $q$ marks $q_x$ and $q_y$ in its marking table.
	\EndFor
	\If{$(1/2-\gamma)$-fraction of parties in any quorum have been marked, for $\gamma = 0.01$}
		\State each of these parties is set unmarked in all its quorums.
		\State each of these parties is set unmarked in all its neighboring quorums.
	\EndIf
	\For{each pair of leaders, $(q_x, q_y)$, that is in conflict, in quorums $(Q_x, Q_y)$}
		\State ELECT($Q_x$) and ELECT($Q_y$) to elect a pair of unmarked leaders, $(q'_x,q'_y)$.
		\State $Q_x$ and $Q_y$ notify their neighboring quorums with $(q'_x,q'_y)$.
	\EndFor
\end{algorithmic}
* A pair of parties, $(q_x, q_y)$, is \emph{in conflict} if: 1) $q_x$ was scheduled to send an output to $q_y$ at some point in the last call to \computecircuit or \advancedcheck; and 2) $q_y$ does not receive an expected message from $q_x$ in \investigate, or $q_y$ receives a message in \investigate that is different than the message that it has received from $q_x$ in the last call to \computecircuit or \advancedcheck.
\end{algorithm}

\section{Analysis}\label{sec:analysis}
In this section, we prove Theorem~\ref{thm:corruptions}. All logarithms are base 2.
\begin{definition}
\emph{Rooted Directed Acyclic Graph (R-DAG)} is a DAG in which, for a vertex $u$ called the root and any other node $v$, there is at least one directed path from $v$ to $u$.  
\end{definition}

\begin{lemma}\label{lem:shrinkingDeceptionDAG}
Given any R-DAG, of size $n$, in which each node has indegree of at most $d$ and survives independently with probability at most $p$ such that $0 < p \leq \frac{1}{d} - \epsilon$, for any constant $\epsilon > 0$, then the probability of having a subgraph, rooted at some node, having only surviving nodes, of size $\Omega(\frac{\log{n}}{(1-pd)^2})$ is at most $1/2$.
\end{lemma}
\begin{proof} 
This proof makes use of the following three propositions, but first we define some notations.

Given an R-DAG, $D(V,E)$, with size $n$ and maximum indegree $d$, after each node survives independently with probability at most $p$ such that $0 < p \leq \frac{1}{d} - \epsilon$, for any constant $\epsilon > 0$, we explore $D$ to find a subgraph with only surviving nodes of size more than $k$, rooted at an arbitrary node $v$ (assuming that node $v$ survives). 

Let $D'(v)$ be the maximum subgraph of surviving nodes, rooted at node $v$.
Let each node in $D$ have a status, which is either \emph{\inactivenode}, \emph{\activenode} or \emph{\neutralnode}. During the exploration process, the status of nodes is changed. A node $x$ is \inactivenode ~if $x \in D'(v)$ and its children are explored determining which one is in $D'(v)$. A node $x$ is \activenode ~if $x \in D'(v)$ but its children are not explored yet. A node $x$ is \neutralnode ~if it is neither \activenode ~nor \inactivenode, i.e., node $x$ and its children are not explored yet.

The exploration process runs in at most $k>0$ steps. Initially, we set an arbitrary surviving node, $v$, \activenode ~and all other nodes \neutralnode.
At each step $i$, we choose an active node, $w_i$, in an arbitrary way, and we explore all its children. For all $(w_i,w'_i) \in E$ and $w'_i$ survives and is neutral, we set $w'_i$ active, otherwise $w'_i$ remains as it is.
Then, we set $w_i$ inactive. 
Note that at any step, if there is no active node, the exploration process terminates.
Now let $d_i$ be the maximum number of children of node $w_i$ for $1\leq i \leq k$, i.e., 
\[d_i = \left\{ 
\begin{array}{l l}
  deg(w_i) - 1 & \quad \mbox{if $w_i \in V - root(D)$,}\\
  deg(w_i) & \quad \mbox{$otherwise$.}\\ 
\end{array} \right.\]
where $deg(w_i)$ is the degree of node $w_i$ and $root(D)$ is the root node of $D$.
For $1\leq i \leq k$, let $X_i$ be a non-negative random variable for the number of surviving neutral children of $w_i$, and let $Y_i$ be a non-negative random variable for the number of surviving non-neutral children of $w_i$. Note that $Y_1 = 0$.
So $X_i$ follows a binomial distribution with parameters $(d_i - Y_i)$ and $p$, i.e., $X_i \sim Bin(d_i - Y_i, p)$.
Let $A_i$ be a non-negative random variable for the total number of active nodes after $i$ steps, for $1\leq i \leq k$.

\begin{proposition}\label{proposition:A_i}
$A_i = \left\{ 
\begin{array}{l l}
  \sum_{i=1}^{k} X_i - (k-1) & \quad \mbox{if $A_{i-1} > 0$,}\\
  0 & \quad \mbox{$otherwise$.}\\ 
\end{array} \right.$
\end{proposition}
\begin{proof}
Since the process starts initially with one active node $v$, $A_0 = 1$.
Now we have two cases of $A_{i-1}$ to compute $A_i$, $1\leq i \leq k$:\\
\textbf{Case 1 (process terminates before $i$ steps):} If $A_{i-1} = 0$, then $A_j = 0$ for $i\leq j\leq k$.\\
\textbf{Case 2 (otherwise):} If $A_{i-1} > 0$, then $A_i = A_{i-1} + X_i - 1$,
where after exploring $w_i$, the total number of active nodes is the number of new active nodes ($X_{i}$) due to the exploration of $w_i$ in addition to the total number of active nodes of previous steps ($A_{i-1}$) excluding $w_i$ that becomes inactive at the end of step $i$.
\qed
\end{proof}
Now let $|D'(v)|$ be the number of nodes in $D'(v)$. 
\begin{proposition}\label{proposition:Pr(|D'(v)| > k)}
$Pr(|D'(v)| > k) \leq Pr(\sum_{i=1}^{k} X_i \geq k).$
\end{proposition}

\begin{proof}
To prove this proposition, we first prove that $Pr(|D'(v)| > k) \leq Pr(A_k > 0).$
In order to do that, we prove that $|D'(v)| > k \implies A_k > 0$. 
If $|D'(v)| > k$, then the exploration process does not terminate before $k$ steps. This implies that after k steps, there are $k$ inactive nodes and at least one active node remains. This follows that $A_k > 0$.
Thus, we have 
$
Pr(|D'(v)| > k) \leq Pr(A_k > 0).
$

Now we prove that $Pr(A_k > 0) \leq Pr(\sum_{i=1}^{k} X_i - (k-1) > 0)$.
To do that, we prove that $A_k > 0 \implies \sum_{i=1}^{k} X_i - (k-1) > 0$.
If $A_k > 0$, then $A_j > 0$ for all $1\leq j\leq k$.
By Proposition \ref{proposition:A_i}, we obtain that $\sum_{i=1}^{j} X_i - (j-1) > 0$ for all $1\leq j\leq k$. This follows that 
$
Pr(A_k > 0) \leq Pr(\sum_{i=1}^{k} X_i - (k-1) > 0).
$

Now, we obtain
$
Pr(|D'(v)| > k) \leq Pr(\sum_{i=1}^{k} X_i - (k-1) > 0),
$
or equivalently, 
$
Pr(|D'(v)| > k) \leq Pr(\sum_{i=1}^{k} X_i > k - 1).
$
Since $k$ is a positive integer, we have 
$
Pr(|D'(v)| > k) \leq Pr(\sum_{i=1}^{k} X_i \geq k).
$
\qed
\end{proof}

\begin{proposition}\label{proposition:sumX_i}
$Pr(\sum_{i=1}^{k} X_i \geq k) \leq e^{-\frac{(1-pd)^2 k}{1+pd}}.$
\end{proposition}
\begin{proof}
To prove this proposition, we first make use of stochastic dominance. 
For $1 \leq i \leq k$, let $X_i^+ \sim Bin(d, p)$, and let $X_1^+, ..., X_k^+$ be independent random variables.
We know that $Y_i \geq 0$ and $d_i \leq d$ for $1 \leq i \leq k$. By Theorem (1.1) part (a) of \cite{Klenke:2010}, for all $1 \leq i \leq k$, $X_i^+$ first-order stochastically dominates $X_i$, i.e.,  $X_i^+$ is stochastically larger than $X_i$.
Hence, $\sum_{i=1}^{k} X_i^+$ is stochastically larger than $\sum_{i=1}^{k} X_i$.
Thus, we have
$
Pr(\sum_{i=1}^{k} X_i \geq k) \leq Pr(\sum_{i=1}^{k} X_i^+ \geq k).
$


Now let $S_k = \sum_{i=1}^{k} X_i^+$. By Chernoff bounds, for $\delta > 0$, we obtain
$$
Pr(S_k \geq (1+\delta)E(S_k))
\leq \left(\frac{e^\delta}{(1+\delta)^{(1+\delta)}}\right)^{E(S_k)}
\leq e^{-\frac{\delta^2}{2+\delta}E(S_k)}.
$$
We know that $S_k \sim Bin(kd, p)$.
Thus, $E(S_k) = pdk$.
Therefore, we have
$$
Pr(S_k \geq (1+\delta)pdk) \leq e^{-\frac{\delta^2}{2+\delta}pdk}.
$$
For $\delta = \frac{1-pd}{pd}$, we obtain
$
Pr(S_k \geq k) \leq e^{-\frac{(1-pd)^2 k}{1+pd}}.
$\qed
\end{proof}

Now by Propositions \ref{proposition:Pr(|D'(v)| > k)} and \ref{proposition:sumX_i}, we have
$
Pr(|D'(v)| > k) \leq e^{-\frac{(1-pd)^2 k}{1+pd}}.
$
We know that node $v$ survives with probability at most $p$. Thus, we obtain
$$
Pr(|D'(v)| > k) \leq p e^{-\frac{(1-pd)^2 k}{1+pd}}.
$$
Union bound over $n$ nodes, then the probability that there exists a subgraph of $D$, rooted at one node, having only surviving nodes of size more than $k$ is at most
$$
n Pr(|D'(v)| > k) \leq n p e^{-\frac{(1-pd)^2k}{1+pd}}.
$$
Note that $n p e^{-\frac{(1-pd)^2k}{1+pd}} \leq 1/2$, when $k \geq \frac{1+pd}{(1-pd)^2 \log{e}} \log(2pn)$. 
Thus, the probability of having such a subgraph of size more than $\frac{1+pd}{(1-pd)^2 \log{e}} \log(2pn)$, or equivalently, $\Omega\left(\frac{\log{n}}{(1-pd)^2}\right)$, is at most $1/2$.\qed
\end{proof}

\begin{corollary}\label{corr:half}
For any R-DAG, of size $n$, the probability of having a subgraph, rooted at one node, having only surviving nodes, of size at least $n/2$ is $o(1)$.
\end{corollary}

Recall that in each round of \advancedcheck, the new added parties are chosen uniformly at random.
We define a \emph{deception DAG}, $D_i$, as a subgraph of the DAG containing the new added parties, that has the following properties: 1) it has only bad parties, of which the adversary makes use to corrupt the computation; 2) it receives all its inputs, and each input is provided correct by at least one good party; and 3) each output it produces is provided to at least one good party, and at least one output is dropped or corrupted.

When the adversary corrupts at least one output of a deception DAG in any round, it has to keep doing such corruption in all subsequent deception DAGs; otherwise, the good parties that expect to receive certain messages from the deception DAG will receive inconsistent messages and so they call \update. Even if the adversary keeps corrupting these outputs from a round to another, the corruption will be detected if the deception DAG shrinks to size zero at any round.

In the following facts, we show that 1) any deception DAG never expands in any direction from a round to another; and
2) it shrinks logarithmically from a round to another with probability at least $1/2$.
This will imply that any deception DAG shrinks to size zero in $\asymchkrows$ rounds with probability at least $1/2$.

%

\begin{fact}\label{fact:supergraph}
The deception DAG in any round is a supergraph of all subsequent deception DAGs.
\end{fact}
%
%

Let $p$ be the probability of selecting an unmarked bad party uniformly at random in any quorum. Recall that the fraction of bad parties in any quorum is at most $1/4$, and at any time the fraction of unmarked parties in any quorum is at least $1/2+\gamma$, for any constant $\gamma>0$.
Thus, the fraction of unmarked bad parties in any quorum is $\frac{1/2}{1+2 \gamma}$, i.e., $p \leq \frac{1/2}{1+2 \gamma}$.

\begin{fact}\label{fact:shrinktozero}
With probability at least $1/2$, any deception DAG, of size $m$, that is rooted at one party, shrinks to size zero in $O(\log^*{m})$ rounds, where $0 < p \leq \frac{1/2}{1+2 \gamma}$, for any constant $\gamma > 0$.
\end{fact}
\begin{proof}
Given a deception DAG, of size $m$, that is rooted at one party.
By Fact \ref{fact:supergraph}, the deception DAG never expands over rounds.
Let $X_i$ be an indicator random variable, which is $1$ if the deception DAG shrinks logarithmically from a round $i$ to round $i+1$; and $0$ otherwise. 
Note that in each round, the receiver that is elected by the output quorum is good with probability at least $1/2$. 
By Lemma \ref{lem:shrinkingDeceptionDAG}, the probability that the deception DAG shrinks logarithmically is at least $1/2$ given that the elected receiver is good. Since the probability of electing a good receiver is independent of the probability that the deception DAG shrinks logarithmically, $X_i = 1$ with probability at least $1/4$.

Now we show the required number of the $X_{i}$ random variables to be $1$ in order to shrink the deception DAG of size $m$ to $0$.
Now fix $p$. Due to the constant factor, $c\leq\frac{3}{(1-pd)^2}$, of the logarithmic shrinking in Lemma \ref{lem:any-deception-dag-shrinks-to-zero}, then after having $(\log^*{m}-2)$ of $X_{i}$'s equal $1$, the deception DAG of size $m$ will shrink to size $2c(\log{2c})$. 
Moreover, Lemma \ref{lem:any-deception-dag-shrinks-to-zero} will not be applicable after the deception DAG shrinks to a constant size, $k\leq c\log{k}$.
For such case, we simply make use of Corollary \ref{corr:half}.

Thus, after having $(\log^*{m}-2)$ of $X_{i}$'s equal $1$, by Corollary \ref{corr:half}, we further require at most $\log({2c(\log{2c}})$, of $X_{i}$'s equal $1$ to eventually shrink the deception DAG to size $0$. This implies that we require at most $2\log^*{m}$ of the $X_{i}$ random variables to be $1$.

Let $X = \sum_{i=1}^{\realchkrows} X_i$. Then $E(X) = 4 \log^* m$, and since the $X_{i}$'s are independent, by Chernoff bounds,
$$
\Pr\left(X < (1-\delta)4\log^*m\right) \leq 
\left(\frac{e^\delta}{(1+\delta)^{1+\delta}}\right)^{4\log^*m}.
$$

For $\delta = \frac{1}{2}$ and $m > 2$, 
$
\Pr\left(X < 2 \log^*m\right) \leq 
\left(\frac{e^\frac{1}{2}}{(\frac{3}{2})^{\frac{3}{2}}}\right)^{4\log^*m} < \frac{1}{2}.
$
\qed
\end{proof}

\begin{lemma}\label{lem:any-deception-dag-shrinks-to-zero}
With probability at least $1/2$, any deception DAG, of size $m$, shrinks to size zero in $\asymchkrows$ rounds, where $0 < p \leq \frac{1/2}{1+2 \gamma}$, for any constant $\gamma > 0$.
\end{lemma}
\begin{proof}
Let $D$ be a deception DAG with multiple outputs and the adversary corrupts more than one output.
Now let $D_m$ be the maximum subgraph of $D$, that is rooted at one party of a corrupted output. Note that by definition any maximal subgraph that is rooted at one party of a corrupted output is a deception DAG. By Fact \ref{fact:supergraph}, each of these deception DAGs never expands over rounds.
By Fact \ref{fact:shrinktozero}, each of these deception DAGs shrinks to size zero in a number of rounds at most the number of rounds that $D_m$ shrinks to size zero.

Now we consider the case where the adversary maintains more than one deception DAG in the same round. By Fact \ref{fact:supergraph} and the definition of deception DAG, the deception DAGs of the same round do not overlap, and they shrink independently. 
Also, by Fact \ref{fact:shrinktozero}, any maximum deception DAG, of one root, shrinks to size zero in a number of rounds that is at least the number of rounds that any other deception DAG, of one root, shrinks to size zero. 

Therefore, for the adversary to maximize the expected number of rounds in \advancedcheck is to consider the maximum deception DAG, of one root, in the first round of \advancedcheck. 
We know that this maximum deception DAG has size at most $\circuitsize$. By Fact \ref{fact:shrinktozero}, it shrinks to size zero in $\asymchkrows$.
\qed
\end{proof}

The next lemma shows that \advancedcheck catches corruptions with probability $\geq 1/2$.

\begin{lemma}\label{l:check}
Assume some party selected uniformly at random in the last call to \computecircuit has corrupted a computation.  Then when the algorithm \advancedcheck is called, with probability at least $1/2$, some party will call \update.
\end{lemma}

\begin{lemma}\label{l:update}
If some party selected uniformly at random in the last call to \computecircuit or \advancedcheck has corrupted a computation, then \update will identify a pair of neighboring quorums $Q$ and $Q'$ such that at least one pair of parties in these quorums is in conflict and at least one party in such pair is bad.
\end{lemma}

The next lemma bounds the number of calls to \update before all bad parties are marked.

\begin{lemma}\label{l:numc}
\update is called $O(t)$ times before all bad parties are marked.
\end{lemma}

\section{Conclusion and Future Work}\label{sec:conclusion}
We have presented algorithms for reliable multiparty computations. 
These algorithms can significantly reduce message cost and number of computational operations to be asymptotically optimal.
The price we pay for this improvement is the possibility of computation corruption.  In particular, if there are $t\le (\frac{1}{4}-\epsilon)n$ bad parties, for any constant $\epsilon > 0$, our algorithm allows $O(t \chkprobinv)$ computations to be corrupted in expectation.  

Many problems remain.  
First, it seems unlikely that the smallest number of corruptions allowable by an attack-resistant algorithm with optimal message complexity is $O(t \chkprobinv)$.  Can we improve this to $O(t)$ or else prove a non-trivial lower bound?  
Second, we allow the inputs of parties to reveal. Can we maintain the privacy of these inputs?
Third, we assume a partially synchronous communication model, which is crucial for our \advancedcheck algorithm to detect computation corruptions over rounds. Can we extend this algorithm to fit for asynchronous computations?

\bibliographystyle{plain}
\bibliography{compute}

\newpage

\appendix
\section{Appendix - Deferred Proofs}\label{app:proofoftheorem}

\noindent
{\bf Fact \ref{fact:supergraph}.} 
The deception DAG in any round is a supergraph of all subsequent deception DAGs.
\begin{proof}
We know by definition that the deception DAG is bordered by the good parties that provide the inputs to the DAG, and the good parties that receive the outputs from the DAG.

In each round, all parties of each subquorum in round $i$ send their outputs to the new added party in the next subquorum. Thus, the good parties that provide the correct inputs to the deception DAG of round $i$, will provide the correct inputs to the deception DAGs in all subsequent rounds.
Moreover, in each round, the new added party in each subquorum forwards its output to all parties in its subquorum.
Note that each good party has previously received a message, it verifies this message with all subsequent messages it receives, and if it receives inconsistent messages or fails to receive an expected message, then it calls \update.

Therefore, all good parties that border a deception DAG in any round will border all subsequent deception DAGs.\qed
\end{proof}

\noindent
{\bf Lemma \ref{l:check}.} 
Assume some party selected uniformly at random in the last call to \computecircuit has corrupted a computation.  Then when \advancedcheck is called, with probability at least $1/2$, some party will call \update.

\begin{proof}
Recall that in \request, in each round for $\asymchkrows$ rounds, a new party $\receiver$ is elected by the output quorum to send a request to the senders in order to recompute. 
By Lemma \ref{lem:any-deception-dag-shrinks-to-zero}, this request is sent reliably with probability at least $1/2$ to all input quorums.
Note that each request has a round number. 
Thus, at any round, if any good party in any input quorum receives a request of round number $i$ and has not received $(i-1)$ requests of proper round numbers, then it calls \update.

If all input quorums receive all requests properly in all $\asymchkrows$ rounds, then \recompute must be called properly $\asymchkrows$ times by all input quorums.
By Lemma \ref{lem:any-deception-dag-shrinks-to-zero}, the result is computed and sent reliably to the output quorum with probability at least $1/2$.

Recall that in \recompute, the round number $i$ is forwarded with the computation results from the senders to the output quorum. Thus, at any round, if any good party in the output quorum receives a result with a round number $i$ and has not received $(i-1)$ results with proper rounds numbers, then it calls \update.

Now, if all parties in the output quorum receive all results properly in all $\asymchkrows$ rounds, then \resendresult must be called $\asymchkrows$ times by the output quorum.
By Lemma \ref{lem:any-deception-dag-shrinks-to-zero}, the result of the computation is sent back reliably to all senders with probability at least $1/2$.
Thus, the probability that \advancedcheck succeeds in finding a corruption and calling \update is at least $1/2$. \qed
\end{proof}

\noindent
{\bf Lemma \ref{l:update}.} 
If some party selected uniformly at random in the last call to \computecircuit or \advancedcheck has corrupted a computation, then \update will identify a pair of neighboring quorums $Q$ and $Q'$ such that at least one pair of parties in these quorums is in conflict and at least one party in such pair is bad.

\begin{proof}
First, we show that if a pair of parties $x$ and $y$ is in conflict, then at least one of them is bad.  Assume not. 
Then both $x$ and $y$ are good.  
This implies that party $x$ would have truthfully reported what it received and sent; any result that $x$ has computed would have been sent directly to $y$; and $y$ would have truthfully reported what it received from $x$.  But this is a contradiction, since for $x$ and $y$ to be in conflict, $y$ must have reported that it received from $x$ something different than what $x$ reported sending.

Now consider the case where a selected unmarked bad leader corrupted the computation in the last call to \computecircuit. By Lemma \ref{l:check}, with probability at least $1/2$, some party, $q' \in Q'$, will call \update.
Recall that in \update $q'$ broadcasts all messages it has received to all parties in $Q'$. These parties verify if $q'$ received inconsistent messages before proceeding.

In \update, we know that each party, $q \in Q$, participated in the last call to \compute broadcasts what it has received and sent to all parties in $Q$. Thus, all parties of $Q$ verify the correctness of $q$'s computation. Thus, if the corruption occurs due to an incorrect computation made by a bad party, this corruption will be detected and all parties will know that this party is bad.

Now if all parties compute correctly and \advancedcheck detects a corruption, then we show that there is some pair of parties will be in conflict.
Assume this is not the case. 
Thus, by the definition of corruption, there must be a deception DAG, in which all inputs are provided correct and an output is corrupted at party $q'$.
Then each pair of parties, $(q_{j}, q_{k}) \in (Q_{j}, Q_{k})$, in the deception DAG that is rooted at $q'$, is not in conflict, for $n+1\leq j<k \leq m+n$. 
Thus, we have that 1) this DAG received all its inputs correct; 2) all parties compute correctly; and 3) no pair of parties is in conflict. This implies that it must be the case that $q'$ received the correct output.
But if this is the case, then $q'$ that initially called \update would have received no inconsistent messages.  This is a contradiction since in such a case, this party would have been unsuccessful in trying to initiate a call to \update.  
Thus, \update will find two parties that are in conflict, and at least one of them will be bad.
\qed
\end{proof}

\noindent
{\bf Lemma \ref{l:numc}.} 
\update is called $O(t)$ times before all bad parties are marked.
\begin{proof}
By Lemma~\ref{l:update}, if a corruption occurred in the last call to \computecircuit, and it is caught by \advancedcheck, then \update is called. \update identifies at least one pair of parties that is in conflict, and each of such pairs has at least one bad party.

Now let $g$ be the number of marked good parties, and let $b$ be the number of marked bad parties.
Also let $f(b, g) = b - (\frac{p}{1-p})g$. Since $0< p \leq \frac{1/2}{1+2\gamma}$, for any constant $\gamma > 0$, $0< \frac{p}{1-p} \leq \frac{1}{1+4\gamma}$.

For each corruption caught,
at least one bad party is marked, and so $f(b, g)$ increases by at least $\frac{1-2p}{1-p}$ since $b$ increases by at least $1$ and $g$ increases by at most $1$.
When $(1/2-\gamma)$-fraction of parties in any quorum $Q$ get unmarked, for any constant $\gamma > 0$, 
$f(b, g)$ further increases by at least $0$ since $b$ decreases by at most $p|Q|(1/2-\gamma)$ and $g$ decreases by at least $(1-p)|Q|(1/2-\gamma)$.
Hence, $f(b, g)$ is monotonically increasing by at least $\frac{1-2p}{1-p}$ for each corruption caught.
When all bad parties are marked, $f(b, g) \leq t$.
Therefore, after at most $(\frac{1-p}{1-2p})t$, or at most $(1+\frac{1}{2\gamma})t$, calls to \update, all bad parties are marked.\qed
\end{proof}

\noindent
{\bf Proof of Theorem \ref{thm:corruptions}.}
We first show the message cost, the number of operations and the latency of our algorithms.
By Lemma \ref{l:numc}, the number of calling \update is at most $O(t)$. Thus, the resource cost of all calls to \update is bounded as the number of calls to \compute grows large.
Therefore, for the amortized cost, we consider only the cost of the calls to \computecircuit and \advancedcheck.

When a computation is performed through a circuit of $m$ gates with a circuit depth $\ell$, \computecircuit has message cost $O(m + n\log{n})$, number of operations $O(m + n\log{n})$ and latency $O(\ell)$.
\advancedcheck has message cost $O((m + n \log{n})\chkprobinv)$, number of operations $O((m+n\log{n})\logstarcircuitsize)$ and latency $O(\ell \logstarcircuitsize)$, but \advancedcheck is called only with probability $1 / \chkprobinv$. 
Hence, the call to \advancedcheck has an amortized expected message cost $O(m + n \log{n})$, amortized computational operations $O(\frac{m+n\log{n}}{\logstarcircuitsize})$ and an amortized expected latency $O(\ell/ \logstarcircuitsize)$.

In particular, if we call \compute $\mathcal{L}$ times, then the expected total number of messages sent will be $O(\mathcal{L} (m + n\log n) + t(m \log^{2} n))$ with expected total number of computational operations $O(\mathcal{L} (m + n\log n) + t(m\log{n}\log^*{m}))$ and latency $O(\ell(\mathcal{L}+t))$.  This is true since \update is called $O(t)$ times and each call to \update has message cost $O(m\log^2{n})$ with computational operations $O(m\log{n}\log^*{m})$ and latency $O(\ell)$.

Recall that by Lemma~\ref{l:numc}, the number of times \advancedcheck must catch corruptions before all bad parties are marked is $O(t)$.  
In addition, if a bad party caused a corruption during a call to \computecircuit, then by Lemmas~\ref{l:check} and~\ref{l:update}, with probability at least $1/2$, \advancedcheck will catch it.  
As a consequence, it will call \update, which marks the parties that are in conflict.  $\update$ is thus called with probability $1/\chkprobinv$, so the expected total number of corruptions is $O(t \chkprobinv)$.

\end{document}